\documentclass[11pt,reqno]{article}

\usepackage{amssymb,amscd,amsmath,amsthm}
\numberwithin{equation}{section}
\usepackage{graphicx}
\newtheorem{theorem}{Theorem}[section]

\newtheorem{lemma}[theorem]{Lemma}
\newtheorem{corollary}[theorem]{Corollary}
\newtheorem{definition}[theorem]{Definition}
\newtheorem{remark}[theorem]{Remark}

\def\cb{{\mathcal B}}

\def\ce{{\mathcal E}}

\def\cw{{\mathcal W}}

\def\bc{{\mathbb C}}
\def\bh{{\mathbb H}}

\def\bn{{\mathbb N}}

\def\br{{\mathbb R}}

\def\bz{{\mathbb Z}}

\def\a{\alpha}
\def\b{\beta}

\def\tr{{\rm Tr}}
\def\L{\Lambda}
\def\G{\Gamma}
\def\ce{\mathcal E}

\def\ffi{\varphi}

\def\<{\langle}
\def\>{\rangle}

\def\1{\mathbf{1}}

\def\cw{\cal W}

\def\cal{\mathcal}

\def\ssub{\subset \subset}

\def\s{\sigma}
\def\bh{\mathbf{h}}

\def\id{{\bf 1}\!\!{\rm I}}

\begin{document}

\begin{center}
{\Large {\bf On Quantum Markov Chains on Cayley tree I:\\[2mm]
 uniqueness of the associated chain with $XY$-model on the Cayley tree of order two}}\\[1cm]
\end{center}

\begin{center}
{\large {\sc Luigi Accardi}}\\[2mm]
\textit{Centro Interdisciplinare Vito Volterra\\
II Universit\`{a} di Roma ``Tor Vergata''\\
Via Columbia 2, 00133 Roma, Italy} \\
E-email: {\tt accardi@volterra.uniroma2.it}
\end{center}

\begin{center}
{\large {\sc Farrukh Mukhamedov}}\\[2mm]
\textit{ Department of Computational \& Theoretical Sciences,\\
Faculty of Science, International Islamic University Malaysia,\\
P.O. Box, 141, 25710, Kuantan, Pahang, Malaysia}\\
E-mail: {\tt far75m@yandex.ru, \ farrukh\_m@iiu.edu.my}
\end{center}

\begin{center}
{\large{\sc Mansoor Saburov}}\\[2mm]
\textit{Department of Computational \& Theoretical Sciences,\\
Faculty of Science, International Islamic University Malaysia,\\
P.O. Box, 141, 25710, Kuantan, Pahang, Malaysia}\\
E-mail: {\tt msaburov@gmail.com}
\end{center}

\begin{abstract}
In the present paper we study forward Quantum Markov Chains (QMC)
defined on Cayley tree. A construction of such QMC is provided,
namely we construct states on finite volumes with boundary
conditions, and define QMC as a weak limit of those states which
depends on the boundary conditions. Using the provided construction
we investigate QMC associated with $XY$-model on a Caylay tree of
order two. We prove uniqueness of QMC associated with such a model,
this means the QMC does not depend on the boundary conditions.

\end{abstract}

\section{Introduction }\label{intr}

Nowadays, it is know that Markov fields play an important role in
classical probability, in physics, in biological and neurological
models and in an increasing number of technological problems such as
image recognition. Therefore, it is quite natural to forecast that
the quantum analogue of these models will also play a relevant role.
The quantum analogues of Markov processes were first constructed in
\cite{[Ac74f]}, where the notion of quantum Markov chain (QMC) on
infinite tensor product algebras was introduced. Nowadays, QMC have
become a standard computational tool in solid state physics, and
several natural applications have emerged in quantum statistical
mechanics and quantum information theory. The reader is referred to
\cite{fannes2,G,ILW,Kum,OP} and the references cited therein, for
recent developments of the theory and the applications.

A first attempts to construct a quantum analogue of classical Markov
fields has been done in \cite{[Liebs99]}, \cite{[AcFi01a]},
\cite{[AcFi01b]},\cite{AcLi}. These papers extend to fields the
notion of {\it quantum Markov state} introduced in \cite{[AcFr80]}
as a sub--class of the QMC introduced in \cite{[Ac74f]}. In
\cite{[AcFiMu07]} it has been proposed a definition of quantum
Markov states and chains, which extend a proposed one in
\cite{Oh05}, and includes all the presently known examples.   Note
that in the mentioned papers quantum Markov fields were considered
over multidimensional integer lattice $\bz^d$. This lattice has so
called amenability condition. Therefore, it is natural to
investigate quantum Markov fields over non-amenable lattices. One of
the simplest non-amenable lattice is a Cayley tree. First attempts
to investigate QMC over such trees was done in \cite{aklt}, such
studies were related to the investigation of thermodynamic limit of
valence-bond-solid models on a Cayley tree \cite{fannes}.  There, it
was constructed finitely correlated states as ground states of
VBS-model on Cayley tree. The mentioned considerations naturally
suggest the study of the following problem: the extension to fields
the notion of QMC. In \cite{AOM} we have introduced a hierarchy of
notions of Markovianity for states on discrete infinite tensor
products of $C^*$--algebras and for each of these notions we
constructed some explicit examples. We showed that the construction
of \cite{[AcFr80]} can be generalized to trees. It is worth to note
that, in a different context and for quite different purposes, the
special role of trees was already emphasized in \cite{[Liebs99]}.
Noncommutative extensions of classical Markov fields, associated
with Ising and Potts models on Cayley tree, were investigated in
\cite{Mukh04,MR04}. In the classical case, Markov fields on trees
were also considered in \cite{[Pr]}-\cite{[Za85]}.

In the present paper we continue our investigations started in
\cite{AOM}. In \cite{AOM} we have studied backward QMC defined on
the Cayley tree.
 Note that shift
invariant backward QMC chains can be also considered as an extension
of $C^*$-finitely correlated states defined in \cite{fannes2} to the
Cayley trees. But the forward QMC cannot be described by the
finitely correlated ones (see Remark 3.4 below). Therefore, in
section 3 we provide a construction of forward QMC. Namely we
construct states on finite volumes with boundary conditions, and
define QMC as a weak limit of those states which depends on the
boundary conditions. There, we involve some methods used in the
theory of Gibbs measures on trees (see \cite{Geor}). Such
constructions extend ones provided in \cite{Ac87,[AcFr80]}. In
section 4, by means of the provided construction we investigate QMC
associated with $XY$-model on a Cayley tree of order two. For that
model, in a QMC scheme, we prove uniqueness of the limiting state,
i.e. which does not depend on the boundary conditions. Note that
whether or not the resulting states have a physical interest is a
question that cannot be solved on a purely mathematical ground. We
have to stress that  classical $XY$-model have been investigated by
many authors on a 1D-lattice \cite{MH,YTC}, and also on a Cayley
tree \cite{BO}. In a quantum setting such a model were studied in
\cite{arm,LSM,K,FH}.

\section{Preliminaries}\label{dfqmf}

Recall that a Cayley tree $\Gamma^k$ of order $k \ge 1$ is an
infinite tree whose each vertices have exactly $k+1$ edges. The
vertices $x$ and $y$ are called {\it nearest neighbors} and they are
denoted by $l=<x,y>$ if there exists an edge connecting them. A
collection of the pairs $<x,x_1>,\dots,<x_{d-1},y>$ is called a {\it
path} from the point $x$ to the point $y$. The distance $d(x,y),
x,y\in V$, on the Cayley tree, is the length of the shortest path
from $x$ to $y$. If we cut away  an edge $\{x,y\}$ of the tree
$\Gamma^k$, then $\Gamma^k$ splits into connected components, called
semi-infinite trees with roots $x$ and $y$, which will be denoted
respectively by $\Gamma^k(x)$ and $\Gamma^k(y)$. If we cut away from
$\Gamma^k$ the origin $O$ together with all $k+1$ nearest neighbor
vertices, in the result we obtain $k$ semi-infinite $\Gamma^k(x)$
trees with $x \in S_0 = \{ y \in \Gamma^k \, : \, d(O ,y) =1\}$.
Hence we have
\[
\Gamma^k = \bigcup_{x\in S_0} \Gamma^k(x) \cup \{ O\}.
\]

Therefore, in the sequel we will consider semi-infinite Cayley tree
$\Gamma^k = (L,E)$ with the root $x^0$, $L$ is the set of vertices
and $E$ is the set of edges.

Now we are going to introduce a coordinate structure in $\G^k$ as
follows: every vertex $x$ (except for $x^0$) of $\G^k$ has
coordinates $(i_1,\dots,i_n)$, here $i_m\in\{1,\dots,k\}$, $1\leq
m\leq n$ and for the vertex $x^0$ we put $(0)$.  Namely, the symbol
$(0)$ constitutes level 0, and the sites $(i_1,\dots,i_n)$ form
level $n$ ( i.e. $d(x^0,x)=n$) of the lattice (see Fig. \ref{fig1}).

\begin{figure}
\begin{center}
\includegraphics[width=10.07cm]{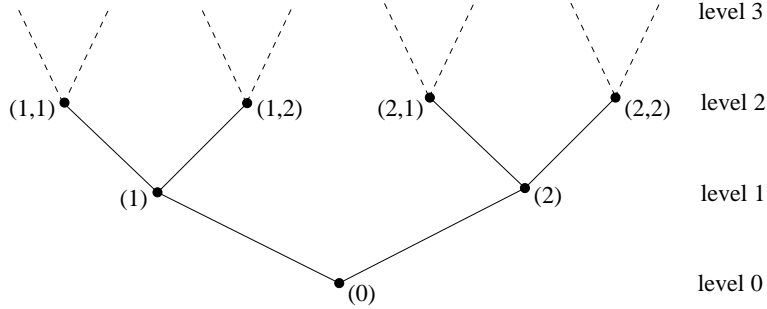}
\end{center}
\caption{The first levels of $\G^2$} \label{fig1}
\end{figure}

Let us set
\[
W_n = \{ x\in L \, : \, d(x,x_0) = n\} , \qquad \Lambda_n =
\bigcup_{k=0}^n W_k, \qquad  \L_{[n,m]}=\bigcup_{k=n}^mW_k, \ (n<m)
\]
\[
E_n = \big\{ <x,y> \in E \, : \, x,y \in \Lambda_n\big\}, \qquad
\Lambda_n^c = \bigcup_{k=n}^\infty W_k
\]
For $x\in \G^k$, $x=(i_1,\dots,i_n)$ denote $ S(x)=\{(x,i):\ 1\leq
i\leq k\}$, here $(x,i)$ means that $(i_1,\dots,i_n,i)$. This set is
called a set of {\it direct successors} of $x$.

 %(see Fig. \ref{fig1}).
%\begin{figure}
%\begin{center}
%\includegraphics[width=10.07cm]{tree.ps}
%\end{center}
%\caption{The first levels of $\G_+^2$} \label{fig1}
%\end{figure}

The algebra of observables $\cb_x$ for any single site $x\in L$ will
be taken as the algebra $M_d$ of the complex $d\times d$ matrices.
The algebra of observables localized in the finite volume $\L\subset
L$ is then given by $\cb_\L=\bigotimes\limits_{x\in\L}\cb_x$. As
usual if $\L^1\subset\L^2\subset L$, then $\cb_{\L^1}$ is identified
as a subalgebra of $\cb_{\L^2}$ by tensoring with units matrices on
the sites $x\in\L^2\setminus\L^1$. Note that, in the sequel, by
$\cb_{\L,+}$ we denote positive part of $\cb_\L$. The full algebra
$\cb_L$ of the tree is obtained in the usual manner by an inductive
limit
$$
\cb_L=\overline{\bigcup\limits_{\L_n}\cb_{\L_n}}.
$$

In what follows, by ${\cal S}({\cal B}_\L)$ we will denote the set
of all states defined on the algebra ${\cal B}_\L$.

Consider a triplet ${\cal C} \subset {\cal B} \subset {\cal A}$ of
unital $C^*$-algebras. Recall that a {\it quasi-conditional
expectation} with respect to the given triplet is a completely
positive (CP) identity preserving linear map $\ce \,:\, {\cal A} \to
{\cal B}$ such that $\ce(ca) = c \ce(a)$, $a\in {\cal A}$, $c \in
{\cal C}$.

\begin{definition}[\cite{AOM}]\label{QMCdef}
Let $\varphi$ be a state on ${\cal B}_L$. Then $\ffi$ is called
\begin{enumerate}
\item[(i)]a {\it forward quantum $d$-Markov chain (QMC)}, associated to
$\{\L_n\}$, if for each $\Lambda_n$, there exist a quasi-conditional
expectation $\ce_{\Lambda_n^c}$ with respect to the triplet
\begin{equation}\label{trplt1}
{\cal B}_{{\Lambda}_{n+1}^c}\subseteq {\cal
B}_{\Lambda_n^c}\subseteq{\cal B}_{\Lambda_{n-1}^c}
\end{equation}
and a state $ \hat\varphi_{\Lambda_n^c}\in{\cal S}({\cal
B}_{\Lambda_n^c}) $ such that for any $n\in {\mathbb N}$ one has
\begin{equation}\label{eq4.1re}
\hat\varphi_{\Lambda_n^c}| {\cal
B}_{\Lambda_{n+1}\backslash\Lambda_n} =
\hat\varphi_{\Lambda_{n+1}^c}\circ \ce_{\Lambda_{n+1}^c}| {\cal
B}_{\Lambda_{n+1}\backslash\Lambda_n}
\end{equation}
and
\begin{equation}\label{dfgqmf}
\varphi = \lim_{n\to\infty} \hat\varphi_{\Lambda_n^c}\circ
\ce_{\Lambda_n^c}\circ \ce_{\Lambda_{n-1}^c} \circ \cdots \circ
\ce_{\Lambda_1^c}
\end{equation}
in the weak-* topology.

\item[(ii)] a {\it backward quantum d-Markov chain}, associated to $\{\L_n\}$, if there
exist a quasi-conditional expectation $\ce_{\Lambda_n}$ with respect
to the triple ${\cal B}_{{\Lambda}_{n-1}}\subseteq {\cal
B}_{\Lambda_n}\subseteq{\cal B}_{\Lambda_{n+1}}$ for each
$n\in\bn$and an initial state $\rho_0\in S(B_{\L_0})$ such that
\begin{equation*}
\varphi = \lim_{n\to\infty} \rho_0\circ \ce_{\Lambda_0}\circ
\ce_{\Lambda_{1}} \circ \cdots \circ \ce_{\Lambda_n}
\end{equation*}
in the weak-* topology.

\end{enumerate}
\end{definition}

In this definition, a forward QMC $\varphi$ generated by ${\cal
E}_{\Lambda_n^c}$ and $\varphi_{\Lambda_n^c}$, is well-defined.
Indeed, we have
\[
\hat\varphi_{\Lambda_n^c} \circ \ce_{\Lambda_{n}^c}
|\cb_{\Lambda_n}= \hat\varphi_{\Lambda_{n+1}^c}\circ
\ce_{\Lambda_{n+1}^c} \circ \ce_{\Lambda_n^c}| {\cal B}_{\Lambda_n}
\]
by (\ref{eq4.1re}) and a following remark so that, for $\Lambda
\ssub \Lambda_k$ and $a\in {\cal B}_{\Lambda}$,
\[
\lim_{n\to\infty}  \hat\varphi_{\Lambda_n^c}\circ
\ce_{\Lambda_n^c}\circ \ce_{\Lambda_{n-1}^c} \circ \cdots \circ
\ce_{\Lambda_1^c}(a) = \hat\varphi_{\Lambda_k^c}\circ
\ce_{\Lambda_k^c}\circ \ce_{\Lambda_{k-1}^c} \circ \cdots \circ
\ce_{\Lambda_1^c}(a).
\]

Similarly, one can also demonstrate that backward QMC is
well-defined.

\begin{remark} Note that in \cite{AOM} a forward QMC was called a generalized quantum Markov state.
\end{remark}

\begin{remark} We have to stress that in most well known papers (see for example
\cite{[AcFi01a], AF03,fannes,fannes2,Lu})
 related to QMC, all such
states were investigated as a backward QMC. Therefore, in the sequel
we will be interested in forward QMC, which is less studied.
\end{remark}

%%%%%%%%%%%%%%%%%%%%%%%%%%%%%%%%%%%%%%%%%%%%%%%%%%%%%%%%%%%% section 5

\section{A constructions of the forward QMC on the Cayley tree}\label{dfcayley}

In this section, we are going to provide a construction of forward
quantum $d$-Markov chain. Note that a construction of backward QMC
has been studied in \cite{AOM}.

Let us rewrite the elements of $W_n$ in the following order, i.e.
\begin{eqnarray*}
\overrightarrow{W_n}:=\left(x^{(1)}_{W_n},x^{(2)}_{W_n},\cdots,x^{(|W_n|)}_{W_n}\right),\quad \overleftarrow{W_n}:=\left(x^{(|W_n|)}_{W_n},x^{(|W_n|-1)}_{W_n},\cdots, x^{(1)}_{W_n}\right).
\end{eqnarray*}
Note that $|W_n|=k^n$. Vertices $x^{(1)}_{W_n},x^{(2)}_{W_n},\cdots,x^{(|W_n|)}_{W_n}$ of $W_n$ can be
represented in terms of the coordinate system as follows
\begin{eqnarray*}
&&x^{(1)}_{W_n}=(1,1,\cdots,1,1), \quad x^{(2)}_{W_n}=(1,1,\cdots,1,2), \ \ \cdots \quad x^{(k)}_{W_n}=(1,1,\cdots,1,k,),\\
&&x^{(k+1)}_{W_n}=(1,1,\cdots,2,1), \quad x^{(2)}_{W_n}=(1,1,\cdots,2,2), \ \ \cdots \quad x^{(2k)}_{W_n}=(1,1,\cdots,2,k),
\end{eqnarray*}
\[\vdots\]
\begin{eqnarray*}
&&x^{(|W_n|-k+1)}_{W_n}=(k,k,,\cdots,k,1), \ x^{(|W_n|-k+2)}_{W_n}=(k,k,\cdots,k,2),\ \ \cdots  x^{|W_n|}_{W_n}=(k,k,\cdots,k,k).
\end{eqnarray*}

Analogously, for a given vertex $x,$ we shall use the following notation for
the set of direct successors of $x$:
\begin{eqnarray*}
\overrightarrow{S(x)}:=\left((x,1),(x,2),\cdots (x,k)\right),\quad
\overleftarrow{S(x)}:=\left((x,k),(x,k-1),\cdots (x,1)\right).
\end{eqnarray*}
In what follows, for the sake of simplicity, we will use notation
$i\in \overrightarrow{S(x)}$ (resp. $i\in \overleftarrow{S(x)}$
instead of $(x,i)\in \overrightarrow{S(x)}$ (resp. $(x,i)\in
\overleftarrow{S(x)}$).

Assume that for each edge $<x,y>\in E$ of the tree an operator
$K_{<x,y>}\in {\cal B}_{\{x,y\}}$ is assigned. We would like to
define a state on $\cb_{\L_n}$ with boundary conditions $w_{0}\in
{\cal B}_{(0),+}$ and $\bh=\{h_x\in {\cal B}_{x,+}\}_{x\in L}$.
To do this, we denote
\begin{eqnarray}\label{Kn1}
K_{[m-1,m]}&:=&\prod_{x\in
\overrightarrow{W}_{m-1}}\prod_{y\in \overrightarrow{S(x)}}K_{<x,y>},\\
\label{Kn2}
\bh^{1/2}_n&:=&\prod_{x\in \overrightarrow{W}_n}h_x^{1/2}, \quad \quad \bh_n:=\bh^{1/2}_n(\bh^{1/2}_n)^{*},\\
\label{Kn3}
K_n&:=&w_0^{1/2}K_{[0,1]}K_{[1,2]}\cdots K_{[n-1,n]}\bh^{1/2}_n,\\
\label{Kn4}
{\cw}_{n]}&:=&K_nK_n^{*},
\end{eqnarray}
It is clear that ${\cw}_{n]}$ is positive.

In what follows, by $\tr_{\L}:\cb_L\to\cb_{\L}$ we mean normalized
partial trace, for any $\Lambda\subseteq_{\text{fin}}L$. For the
sake of shortness we put $\tr_{n]} := \tr_{\Lambda_n}$.

Let us define a positive functional $\ffi^{(n,f)}_{w_0,\bh}$ on
$\cb_{\Lambda_n}$ by
\begin{eqnarray}\label{ffi-ff}
\ffi^{(n,f)}_{w_0,\bh}(a)=\tr(\cw_{n+1]}(a\otimes\id_{W_{n+1}})),
\end{eqnarray}
for every $a\in \cb_{\Lambda_n}$, where
$\id_{W_n+1}=\bigotimes\limits_{y\in W_{n+1}}\id$. Note that here,
$\tr$ is a normalized trace on ${\cal B}_L$.

To get an infinite-volume state $\ffi^{(f)}$ on $\cb_L$  such
that $\ffi^{(f)}\lceil_{\cb_{\L_n}}=\ffi^{(n,f)}_{w_0,\bh}$, we
need to impose some constrains to the boundary conditions
$\big\{w_0,\bh\big\}$ so that the functionals
$\{\ffi^{(n,f)}_{w_0,\bh}\}$ satisfy the compatibility condition,
i.e.
\begin{eqnarray}\label{compatibility}
\ffi^{(n+1,f)}_{w_0,\bh}\lceil_{\cb_{\L_n}}=\ffi^{(n,f)}_{w_0,\bh}.
\end{eqnarray}

\begin{theorem}\label{compa} Let the boundary conditions $w_{0}\in {\cal
B}_{(0),+}$ and ${\bh}=\{h_x\in {\cal B}_{x,+}\}_{x\in L}$ satisfy
the following conditions:
\begin{eqnarray}\label{eq1}
&& \tr ( w_0 h_0 ) =1 \\
\label{eq2} &&\tr_{x]}\left[\prod_{y\in
\overrightarrow{S(x)}}K_{<x,y>} \prod_{y\in
\overrightarrow{S(x)}}h_{y}\prod_{y\in
\overleftarrow{S(x)}}K_{<x,y>}^*\right]=h_x \ \ \textrm{for every} \
\  x\in L.
\end{eqnarray}
Then the functionals $\{\ffi^{(n,f)}_{w_0,\bh}\}$ satisfy the
compatibility condition \eqref{compatibility}. Moreover, there is a
unique forward quantum $d$-Markov chain  $\ffi^{(b)}_{w_0,{\bh}}$ on $\cb_L$ such that
$\ffi^{(f)}_{w_0,{\bh}}=w-\lim_{n\to\infty}\ffi^{(n,f)}_{w_0,\bh}$.
\end{theorem}

\begin{proof} First note that \cite{[AcFr80]} a family of states
$\{\ffi^{(n,f)}_{w_0,\bh}\}$ satisfy the compatibility condition if
a sequence $\{\cw_{n]}\}$ is {\it projective} with respect to
$\tr_{n]}$, i.e.
\begin{equation}\label{pro}
\tr_{n-1]}(\cw_{n]})=\cw_{n-1]}, \ \ \forall n\in\bn.
\end{equation}

Now let us check the equality \eqref{pro}. From \eqref{Kn1}-\eqref{Kn4} one has
\begin{eqnarray*}
\cw_{n]}&=&
w_0^{1/2}\bigg(\prod_{m=1}^{n-1}K_{[m-1,m]}\bigg)K_{[n-1,n]}\bh_{n}K_{[n-1,n]}^*
\bigg(\prod_{m=1}^{n-1}K_{[m-1,m]}\bigg)^* w_0^{1/2}.
\end{eqnarray*}

We know that for different $x$ and $x'$ taken from $W_{n-1}$ the
algebras $\cb_{x\cup S(x)}$ and  $\cb_{x'\cup S(x')}$ commute,
therefore from \eqref{Kn2} one finds
\begin{eqnarray*}
K_{[n-1,n]}\bh_{n}K_{[n-1,n]}^*= \prod_{x\in
\overrightarrow{W}_{n-1}}\bigg(\prod_{y\in\overrightarrow{S(x)}}
K_{<x,y>}\bigg)\bigg( \prod_{y\in\overrightarrow{S(x)}}
h_{y}\bigg)\bigg( \prod_{y\in\overleftarrow{S(x)}}K^*_{<x,y>}\bigg).
\end{eqnarray*}

Hence, from the last equality with \eqref{eq2} we get
\begin{eqnarray*}
\tr_{n-1]}(\cw_{n]})&=&
w_0^{1/2}\bigg(\prod_{m=1}^{n-1}K_{[m-1,m]}\bigg)\\
&&\times \prod_{x\in
\overrightarrow{W}_{n-1}}\tr_x\bigg(\prod_{y\in\overrightarrow{S(x)}}
K_{<x,y>}\prod_{y\in\overrightarrow{S(x)}} h_{y}
\prod_{y\in\overleftarrow{S(x)}}K^*_{<x,y>}\bigg)\\
&&\bigg(\prod_{m=1}^{n-1}K_{[m-1,m]}\bigg)^* w_0^{1/2}
\\
&=&w_0^{1/2}\bigg(\prod_{m=1}^{n-1}K_{[m-1,m]}\bigg)\prod_{x\in
\overrightarrow{W}_{n-1}}h_x\bigg(\prod_{m=1}^{n-1}K_{[m-1,m]}\bigg)^*
w_0^{1/2}\\
&=&\cw_{n-1]}
\end{eqnarray*}

From the above argument and (\ref{eq1}), one can show that
$\cw_{n]}$ is density operator, i.e. $\tr(\cw_{n]})=1$.

Let us show that the defined state $\ffi^{(f)}_{w_0,{\bh}}$ is a
forward QMC. Indeed, define quasi-conditional expectations
$\ce_{\L_n^c}$ as follows:
\begin{eqnarray}\label{E-n1}
&&\hat\ce_{\L_1^c}(x_{[0})=\tr_{[1}(K_{[0,1]}w_0^{1/2}x_{[0}w_0^{1/2}K_{[0,1]}^*), \ \ x_{[0}\in \cb_{\L_0^c}\\
&&\label{E-n2}
\ce_{\L_k^c}(x_{[k-1})=\tr_{[n}(K_{[k-1,k]}x_{[k-1}K_{[k-1,k]}^*), \
\ x_{[k-1}\in\cb_{\L_{k-1}^c}, \ \ k=1,2,\dots,n+1,
\end{eqnarray}
here $\tr_{[n}=\tr_{\L_n^c}$. Then for any monomial $a_{\L_1}\otimes
a_{W_2}\otimes\cdots\otimes a_{W_{n}}\otimes\id_{W_{n+1}}$, where
$a_{\L_1}\in\cb_{\L_1},a_{W_k}\in\cb_{W_k}$, ($k=2,\dots,n$), we
have
\begin{eqnarray}\label{E-n2}
\ffi^{(n,f)}_{w_0,\bh}(a_{\L_1}\otimes a_{W_2}\otimes\cdots\otimes
a_{W_{n}})&=&\tr\bigg(\bh_{n+1}K_{[n,n+1]}^*\cdots
K_{[0,1]}^*w_0^{1/2}(a_{\L_1}\otimes a_{W_2}\otimes\cdots\otimes
a_{W_{n}})\nonumber\\
&&w_0^{1/2}K_{[0,1]}\cdots K_{[n,n+1]}\bigg)\nonumber\\
&=&\tr_{[1}\bigg(\bh_{n+1}K_{[n,n+1]}^*\cdots
K^*_{[1,2]}\hat\ce_{\L_1^c}(a_{\L_1})a_{W_2}K_{[1,2]}\nonumber\\
&&\cdots a_{W_{n}}K_{[n,n+1]}\bigg)\nonumber\\
&=&\tr_{[n+1}\big(\bh_{n+1}\ce_{\L_{n+1}^c}\circ\ce_{\L_n^c}\circ\cdots\nonumber\\
&& \ce_{\L_2^c}\circ\hat\ce_{\L_n^c}(a_{\L_1}\otimes
a_{W_2}\otimes\cdots\otimes a_{W_{n}})\big).
\end{eqnarray}

Hence, for any $a\in\L\subset\L_{n+1}$ from \eqref{ffi-ff} with
\eqref{Kn1},\eqref{Kn2} \eqref{E-n1},\eqref{E-n2} one can see that
\begin{equation}\label{qqq}
\ffi^{(n,f)}_{w_0,\bh}(a)=\tr_{[n+1}\big(\bh^{n+1}\ce_{\L_{n+1}^c}\circ\ce_{\L_n^c}\circ\cdots
\ce_{\L_2^c}\circ\hat\ce_{\L_n^c}(a)\big).
\end{equation}
The projectivity of ${\mathcal{W}}_{n]}$ yields the equality
\eqref{eq4.1re} for $\ffi^{(n,f)}_{w_0,\bh}$, therefore, from
\eqref{qqq} we conclude that $\ffi^{(f)}_{w_0,\bh}$ is a forward
QMC.
\end{proof}

\begin{corollary}\label{state^nwithW_n}
If \eqref{eq1},\eqref{eq2} are satisfied then one has
$\ffi^{(n,f)}_{w_0,\bh}(a)=\tr(\cw_{n]}(a))$ for any $a\in
\cb_{\Lambda_n}$.
\end{corollary}

\begin{remark}  Note that if $k=1$ and $h_x=I$ for all $x\in L$, then
we get conditional amplitudes introduced by L.Accardi
\cite{[AcFr80]}.
\end{remark}

Observe that the state $\ffi^{(f)}_{w_0,\bh}$ has a backward
structure. Indeed, let us first define
\begin{equation}\label{T-n1}
T_k[X](y)=\tr_{k]}(K_{[k,k+1]}XK_{[k,k+1]}^*y), \ \
X\in\cb_{\L_{[k,k+1]}},y\in\cb_{W_{k+1}}.
\end{equation}
Then, using Corollary \ref{state^nwithW_n} one finds
\begin{eqnarray}\label{T-n2}
\ffi^{(n,f)}_{w_0,\bh}(a_{\L_0}\otimes a_{W_1}\otimes\cdots\otimes
a_{W_{n}})&=&\tr\bigg(w_0^{1/2}K_{[0,1]}\cdots
K_{[n-1,n]}\bh_{n}K_{[n-1,n]}^*\cdots
K_{[0,1]}^*w_0^{1/2}\nonumber\\
&&(a_{\L_0}\otimes a_{W_1}\otimes\cdots\otimes a_{W_{n}})\bigg)\nonumber\\
&=&\tr\bigg(w_0^{1/2}K_{[0,1]}\cdots K_{[n-2,n-1]}
\tr_{n-1]}\big(K_{[n-1,n]}\bh_{n}K_{[n-1,n]}^*a_{W_n}\big)\nonumber\\
&&K_{[n-2,n-1]}^*a_{W_{n-1}}\cdots
K_{[0,1]}^*a_{W_1}w_0^{1/2}a_{\L_0}\bigg)\nonumber\\
&=&\tr\bigg(w_0^{1/2}K_{[0,1]}\cdots
K_{[n-3,n-2]}T_{n-2}[T_{n-1}[\bh_{n}](a_{W_n})](a_{W_{n-2}})\nonumber\\
&&K_{[n-3,n-2]}a_{W_{n-2}} \cdots
K_{[0,1]}^*a_{W_1}w_0^{1/2}a_{\L_0}\bigg)\nonumber\\
&=&\tr\bigg(w_0^{1/2}T_0[T_1[\cdots[T_{n-1}[\bh_{n}](a_{W_n})]\nonumber\\
&&(a_{W_{n-2}}) \cdots](a_{W_1})]w_0^{1/2}a_{\L_0}\bigg)
\end{eqnarray}

\begin{remark} Formula \eqref{T-n2} reminds the structure of a
backward quantum Markov chain, however there is an important
difference. For any positive $X\in{\cb}_{\L_{[n-1,n]}}$ the maps
$$
a_n\in{\cb}_{W_n}\mapsto \tr_{n-1]}(K_{[n-1,n]}XK_{[n-1,n]}^*a_n)=:
E(a_n)\in{\cb}_{W_{n-1}}
$$
will be in general anti--CP rather than CP (i.e. the map
$E^*(x):=E(x^*)$ is CP) (see \cite{Ac87}). We will show elsewhere
that there is indeed a deep connection between the present
construction and Cecchini's $\lambda$--operator \cite{Cec89},
\cite{CePe91}.
\end{remark}

\begin{remark} Note that the above construction has the advantage to work on arbitrary local
algebras. It generalizes the construction in \cite{Ac87}. Under
additional assumptions, the local structure of the state becomes
more transparent.  It also exhibits a "forward" local structure
which, however, is not backward Markovian in the sense of Definition
\ref{QMCdef}, but rather in the sense of Cecchini \cite{Cec89}. The
duality between a "forward" and "backward" Markovianity, emerging
from \eqref{qqq}, \eqref{T-n2} is a nontrivial quantum extension of
the fact that in a classical framework the two notions are
equivalent and seems to deserve a deeper study.
\end{remark}

%\begin{definition}
%We say that there exists a phase transition for a family of
%operators $\{K_{<x,y>}\}$ if \eqref{eq1} and \eqref{eq2} have at
%least two $(w_0,\{h_x\}_{x\in L})$ and $(\bar
%w_0,\{\bar{h}_x\}_{x\in L})$ solutions such that the corresponding
%Quantum $d$-Markov chains $\ffi_{w_0,\bh}$ and $\ffi_{\bar
%w_0,\bar\bh}$ are disjoint. Otherwise, we say there is no phase
%transition.
%\end{definition}

%%%%%%%%%%%%%%%%%%%%%%%%%%%%%%%%%%%%%%%%%%55     exam1

\section{Forward QMC associated with XY-model}\label{exam1}

In this section, we prove uniqueness of the quantum $d$-Markov chain
associated with $XY$-model on a Cayley tree of order two. In what
follows, we consider a semi-infinite Cayley tree $\G^2=(L,E)$ of
order 2. Our starting $C^{*}$-algebra is the same $\cb_L$ but with
$\cb_{x}=M_{2}(\bc)$ for $x\in L$. By
$\s_x^{(u)},\s_y^{(u)},\s_z^{(u)}$ we denote the Pauli spin
operators for at site $u\in L$. Here
\begin{equation}\label{pauli}
\s_x^{(u)}= \left(
          \begin{array}{cc}
            0 & 1 \\
            1 & 0 \\
          \end{array}
        \right), \quad
\s_y^{(u)}= \left(
          \begin{array}{cc}
            0 & -i \\
            i & 0 \\
          \end{array}
        \right), \quad
\s_z^{(u)}= \left(
          \begin{array}{cc}
            1 & 0 \\
            0 & -1 \\
          \end{array}
        \right).
\end{equation}

%By $e_{ij}^{(x)}$ we denote the standard matrix units of $\cb_{x} =
%M_2(\bc)$.

For every edge $<u,v>\in E$ put
\begin{equation}\label{1Kxy1}
K_{<u,v>}=\exp\{\b H_{<u,v>}\}, \ \ \b>0,
\end{equation}
where
\begin{equation}\label{1Hxy1}
H_{<u,v>}=\frac{1}{2}\big(\s_{x}^{(u)}\s_{x}^{(v)}+\s_{y}^{(u)}\s_{y}^{(v)}\big).
\end{equation}

Now taking into account the following equalities
\begin{eqnarray*}\label{1Hxy2}
&&H_{<u,v>}^{2m}=H_{<u,v>}^2=\frac{1}{2}\big(\id-\s_{z}^{(u)}\s_{z}^{(v)}\big),\
\ \
%\label{1Hxy3}
H_{<u,v>}^{2m-1}=H_{<u,v>}, \ \ \ m\in\bn,
\end{eqnarray*}
one finds
$$K_{<u,v>}=\id+\sinh\beta H_{<u,v>}+(\cosh\beta-1)H^2_{<u,v>}.$$

We are going to describe all solutions $\bh=\{h_x\}$ and $w_0$ of
the equations \eqref{eq1},\eqref{eq2}. Furthermore, we shall assume
that $h_x=h_y$ for every $x,y\in W_n$, $n\in\bn$. Hence, we denote
$h_x^{(n)}:=h_x$, if $x\in W_n$. Now from
\eqref{1Kxy1},\eqref{1Hxy1} one can see that
$K_{<u,u>}=K^{*}_{<u,v>}$, therefore, the equation \eqref{eq2} can
be rewritten as follows
\begin{eqnarray}\label{state}
\tr_x(K_{<x,y>}K_{<x,z>}h^{(n+1)}_yh^{(n+1)}_zK_{<x,z>}K_{<x,y>})&=&h_x^{(n)},
\ \ \textrm{for every} \ x\in L.
\end{eqnarray}

After small calculations the equation \eqref{state} reduces to the
following system of equations
\begin{equation}\label{mainsystem}
\left\{
\begin{array}{r}
   \left(\dfrac{a^{(n+1)}_{11}+a^{(n+1)}_{22}}{2}\right)^2\cosh^4\beta+a^{(n+1)}_{12}a^{(n+1)}_{21}\sinh^2\beta\cosh\beta = a^{(n)}_{11} \\
   a^{(n+1)}_{12}\dfrac{a^{(n+1)}_{11}+a^{(n+1)}_{22}}{2}\sinh\beta\cosh\beta(1+\cosh\beta)= a^{(n)}_{12} \\
   a^{(n+1)}_{21}\dfrac{a^{(n+1)}_{11}+a^{(n+1)}_{22}}{2}\sinh\beta\cosh\beta(1+\cosh\beta)= a^{(n)}_{21} \\
   \left(\dfrac{a^{(n+1)}_{11}+a^{(n+1)}_{22}}{2}\right)^2\cosh^4\beta+a^{(n+1)}_{12}a^{(n+1)}_{21}\sinh^2\beta\cosh\beta = a^{(n)}_{22}
\end{array}
\right.
\end{equation}
here
\begin{equation*}
h_{x}^{(n)}=\left(
          \begin{array}{cc}
            a^{(n)}_{11} & a^{(n)}_{12} \\
            a^{(n)}_{21} & a^{(n)}_{22} \\
          \end{array}
        \right), \quad\quad
h_{y}^{(n+1)}=h_{z}^{(n+1)}=\left(
          \begin{array}{cc}
            a^{(n+1)}_{11} & a^{(n+1)}_{12} \\
            a^{(n+1)}_{21} & a^{(n+1)}_{22} \\
          \end{array}
        \right).
\end{equation*}
From  \eqref{mainsystem} we immediately get that
$a^{(n)}_{11}=a^{(n)}_{22}$ for all $n\in \bn$.

Self-adjointness of $h_x^{(n)}$, i.e.
$\overline{a^{(n)}_{12}}=a^{(n)}_{21},$ for any $n\in \bn$, allows
us to reduce the system \eqref{mainsystem} to
\begin{equation}\label{equationtohxn}
\left\{
\begin{array}{r}
(a^{(n+1)}_{11})^2\cosh^4\beta+|a^{(n+1)}_{12}|^2\sinh^2\beta\cosh\beta
= a^{(n)}_{11}\\
a^{(n+1)}_{12}a^{(n+1)}_{11}\sinh\beta\cosh\beta(1+\cosh\beta) =
a^{(n)}_{12}
\end{array}
\right.
\end{equation}

\begin{remark}\label{positivityofhxn} Note that according to positivity and
invertability of $h_x^{(n)}$ we conclude that
$a_{11}^{(n)}a_{22}^{(n)}>|a_{12}^{(n)}|^2$ for all $n\in\bn.$
\end{remark}

Now we are going to investigate the derive system
\eqref{equationtohxn}. To do this, let us define a mapping
$f:(x,y)\in \br_+\times\bc\to(x{'},y{'})\in\br_+\times\bc$ by

\begin{equation}\label{formofmapf}
\left\{
\begin{array}{r}
(x{'})^2\cosh^4\beta+|y{'}|^2\sinh^2\beta\cosh\beta = x\\
x{'} y{'}\sinh\beta\cosh\beta(1+\cosh\beta) = y,
\end{array}
\right.
\end{equation}
here as before $\beta>0.$

Taking from both sides of the second equation of \eqref{formofmapf}
modules, we get
\begin{equation*}
\left\{
\begin{array}{r}
(x{'})^2\cosh^4\beta+|y{'}|^2\sinh^2\beta\cosh\beta = x\\
x{'} |y{'}|\sinh\beta\cosh\beta(1+\cosh\beta) = |y|.
\end{array}
\right.
\end{equation*}

Therefore, in the sequel we shall consider the following dynamical
system $f:(x,y)\in\br^2_{+}\to(x{'},y{'})\in\br^2_{+}$ given by
\begin{equation}\label{dynsystem}
\left\{
\begin{array}{r}
(x{'})^2\cosh^4\beta+(y{'})^2\sinh^2\beta\cosh\beta = x\\
x{'} y{'}\sinh\beta\cosh\beta(1+\cosh\beta) = y.
\end{array}
\right.
\end{equation}

Furthermore, due Remark \ref{positivityofhxn}, we restrict the
dynamical system \eqref{dynsystem} to the following domain
$$\Delta=\{(x,y)\in \br^2_{+}: x > y\}.$$

Further, we will need the following auxiliary fact:

\begin{lemma}\label{inequality}
If $\beta>0,$ then
$$0<\sinh\beta\cosh\beta(1+\cosh\beta)<\cosh^4\beta.$$
\end{lemma}

The proof is provided in Appendix.

 Let us first find all of the fixed points of \eqref{dynsystem}.

\begin{theorem}
Let $f$ be a dynamical system given by \eqref{dynsystem}. Then the
following assertions hold true:
\begin{enumerate}
\item[(i)] there is a unique fixed point of $f$ in the
domain $\Delta$;

\item[(ii)] the dynamical system $f$ does not have any
$k$ ( $k\ge 2$)  periodic points in the domain $\Delta.$
\end{enumerate}
\end{theorem}

\begin{proof} (i).
Assume that $(x,y)$ is a fixed point, i.e.
\begin{equation}\label{1fix} \left\{
\begin{array}{r}
x^2\cosh^4\beta+y^2\sinh^2\beta\cosh\beta = x\\
x y\sinh\beta\cosh\beta(1+\cosh\beta) = y.
\end{array}
\right.
\end{equation}

Consider two different cases with respect to $y$.

{\sc Case (a).} Let $y=0.$ Then one finds that either $x=0$ or
$x=\frac{1}{\cosh^4\beta}.$ But, only the point
$(\frac{1}{\cosh^4\beta},0)$ belongs to the domain $\Delta.$\\

{\sc Case (b).} Now suppose $y\neq 0.$ Then from \eqref{1fix} one
finds
\begin{eqnarray*}
x=\frac{1}{\sinh\beta\cosh\beta(1+\cosh\beta)},
\end{eqnarray*}
hence, we obtain
\begin{eqnarray*}
y^2\sinh^2\beta\cosh\beta =
\frac{\sinh\beta\cosh\beta(1+\cosh\beta)-\cosh^4\beta}{\sinh^2\beta\cosh^2\beta(1+\cosh\beta)^2}.
\end{eqnarray*}
But, due to  Lemma \ref{inequality},  we infer that
$$\frac{\sinh\beta\cosh\beta(1+\cosh\beta)-\cosh^4\beta}{\sinh^2\beta\cosh^2\beta(1+\cosh\beta)^2}<0$$
which is impossible. Therefore, in this case the dynamical system
does not have any fixed point.

Consequently, the dynamical system  has a unique fixed point which
is equal to $(\frac{1}{\cosh^4\beta}, 0).$

(ii). Now let us turn to study periodic points of the dynamical
system \eqref{dynsystem}. Assume that the system  has a periodic
point $(x^{(0)},y^{(0)})$ with a period of $k\ge 2$ in $\Delta.$
This means that there are points
$$(x^{(0)},y^{(0)}),(x^{(1)},y^{(1)}),\dots,(x^{(k-1)},y^{(k-1)})\in \Delta,$$
such that they satisfy the following equalities
\begin{equation}\label{conforperiodic}
\left\{
\begin{array}{r}
(x^{(i+1)})^2\cosh^4\beta+(y^{(i+1)})^2\sinh^2\beta\cosh\beta = x^{(i)}\\
x^{(i+1)} y^{(i+1)} \sinh\beta\cosh\beta(1+\cosh\beta) = y^{(i)},
\end{array}
\right.
\end{equation}
where $i=\overline{0,k-1},$ i.e.
$f(x^{(i)},y^{(i)})=(x^{(i+1)},y^{(i+1)}),$ with  $x^{(k)}=x^{(0)},$
$y^{(k)}=y^{(0)}.$

Now again consider two different cases with respect to $y^{(0)}$.\\

{\sc Case (a).} Let $y^{(0)}\neq 0.$ Then $x^{(i)},y^{(i)}$ should
be positive for all $i=\overline{0,k-1}.$ Therefore, we have
\begin{eqnarray*}
\frac{x^{(i)}}{y^{(i)}} &=&
\frac{\bigg(\frac{x^{(i+1)}}{y^{(i+1)}}\bigg)^2\cosh^4\beta+\sinh^2\beta\cosh\beta}{\frac{x^{(i+1)}}{y^{(i+1)}}\sinh\beta\cosh\beta(1+\cosh\beta)}\\
&=&
\frac{\cosh^3\beta}{\sinh\beta(1+\cosh\beta)}\cdot\frac{x^{(i+1)}}{y^{(i+1)}}+\frac{\sinh\beta}{1+\cosh\beta}\cdot\frac{y^{(i+1)}}{x^{(i+1)}},
\end{eqnarray*}
where $i=\overline{0,k-1}.$

Due to  $x^{(i)},y^{(i)}>0$ for all $i=\overline{0,k-1},$ we obtain
\begin{eqnarray}\label{ineqperoid}
\frac{x^{(i)}}{y^{(i)}} >
\frac{\cosh^3\beta}{\sinh\beta(1+\cosh\beta)}\cdot\frac{x^{(i+1)}}{y^{(i+1)}},
\end{eqnarray}
for all  $i=\overline{0,k-1}.$

It then follows  from \eqref{ineqperoid} that
\begin{eqnarray*}
\frac{x^{(0)}}{y^{(0)}}>
\left(\frac{\cosh^3\beta}{\sinh\beta(1+\cosh\beta)}\right)^{k}\cdot\frac{x^{(0)}}{y^{(0)}}.
\end{eqnarray*}
But, the last inequality impossible, since Lemma \ref{inequality}
implies
$$\frac{\cosh^3\beta}{\sinh\beta(1+\cosh\beta)}>1.$$
Hence, in this case, the dynamical system \eqref{dynsystem} does
not have any periodic point with $k\geq 2$.\\

{\sc Case (b).} Now suppose that $y^{(0)}=0.$ Since $k\ge 2$ we have
$x^{(0)}\neq \frac{1}{\cosh^4\beta}.$ So, from
\eqref{conforperiodic} we find that $y^{(i)}=0$ for all
$i=\overline{0,k-1}.$ Then again \eqref{conforperiodic} implies that
$$(x^{(i+1)})^2\cosh^4\beta=x^{(i)}, \ \ \ \forall
i=\overline{0,k-1},$$ which means
$$x^{(i+1)}=\frac{1}{\cosh^2\beta}\sqrt{x^{(i)}}, \ \ \ \forall
i=\overline{0,k-1}.$$ Hence, we have
$$x^{(0)}=\frac{1}{\cosh^{4}\beta}\sqrt[\leftroot{-2}\uproot{3}2^{k+1}]{x^{(0)}\cosh^4\beta}.$$
This yields either $x^{(0)}=0$ or $x^{(0)}= \frac{1}{\cosh^4\beta},$
which is a contradiction.
\end{proof}

Now, we would like to write the dynamical system \eqref{dynsystem}
in an explicit form.   To do end, we should solve the system of
equations \eqref{dynsystem} w.r.t. $(x{'},y{'}).$ From
\eqref{dynsystem} we get
\begin{equation*}
\left\{
\begin{array}{r}
(x{'})^2\cosh^4\beta+(y{'})^2\sinh^2\beta\cosh\beta = x\\[2mm]
 (x{'})^2
(y{'})^2\sinh^2\beta\cosh^2\beta(1+\cosh\beta)^2 = y^2.
\end{array}
\right.
\end{equation*}
Letting $(x{'})^2=u$ and $(y{'})^2=v$, one finds
\begin{equation*}
\left\{
\begin{array}{r}
u\cosh^4\beta+v\sinh^2\beta\cosh\beta = x\\
u v \sinh^2\beta\cosh^2\beta(1+\cosh\beta)^2 = y^2.
\end{array}
\right.
\end{equation*}
Then $v$ can be represented by $u$ as follows
\begin{eqnarray}\label{vviau}
v = \frac{x-u\cosh^4\beta}{\sinh^2\beta\cosh\beta}.
\end{eqnarray}
Using this, we obtain the following quadratic equation
\begin{eqnarray*}
\cosh^5\beta(1+\cosh\beta)^2\cdot
u^2-x\cosh\beta(1+\cosh\beta)^2\cdot u+y^2 = 0.
\end{eqnarray*}
Solving  such a equation w.r.t. $u,$ we can find
\begin{eqnarray*}
u_{\pm}=\frac{x\pm\sqrt{x^2-4y^2\cfrac{\cosh^3\beta}{(1+\cosh\beta)^2}}}{2\cosh^4\beta}
\end{eqnarray*}
Then from \eqref{vviau} one gets
$$v_{\pm}=\frac{x\mp\sqrt{x^2-4y^2\cfrac{\cosh^3\beta}{(1+\cosh\beta)^2}}}{2\sinh^2\beta\cosh\beta}.$$

Since the point $(x{'},y{'})$ belongs to the domain $\Delta,$ then
$u$ should be greater than $v.$ Therefore, an explicit form of
$f:\br^2_{+}\to\br^2_{+}$ given by \eqref{dynsystem} is the
following one\\
\begin{eqnarray}\label{maindynsystem}
\left\{
\begin{array}{l}
x^{'} =
\sqrt{\dfrac{x+\sqrt{x^2-4y^2\cfrac{\cosh^3\beta}{(1+\cosh\beta)^2}}}{2\cosh^4\beta}}\\
\\
y^{'} =
\sqrt{\dfrac{x-\sqrt{x^2-4y^2\cfrac{\cosh^3\beta}{(1+\cosh\beta)^2}}}{2\sinh^2\beta\cosh\beta}}.
\end{array}
\right.
\end{eqnarray}
\begin{remark}\label{confordomain}
Note that from \eqref{maindynsystem} one can see that the map $f$ is
well defined if and only if $x$ and $y$ satisfy
\begin{eqnarray}\label{conforxy}
x \ge 2y\sqrt{\frac{\cosh^3\beta}{(1+\cosh\beta)^2}}.
\end{eqnarray}
Moreover, in this case $f$ maps $\Delta$ into itself.
\end{remark}

\begin{lemma}\label{conditionforimage}
Let $f:\Delta\to\Delta$  be the dynamical system given by
\eqref{maindynsystem}. If $x,y$ are positive and satisfy
\eqref{conforxy} then $x{'},y{'}$ are positive and satisfy the
following inequality
\begin{eqnarray}\label{conforx'y'}
\frac{x{'}}{y{'}}<
\frac{\sinh\beta(1+\cosh\beta)}{\cosh^3\beta}\cdot\frac{x}{y}.
\end{eqnarray}
\end{lemma}
\begin{proof} From \eqref{maindynsystem} one can see that if
$x,y$ are positive and satisfy the condition \eqref{conforxy}, then
$x{'},y{'}$ are positive as well. From \eqref{maindynsystem} we find
\begin{eqnarray*}
\frac{x{'}}{y{'}} &=&
\frac{\sinh\beta(1+\cosh\beta)}{\cosh^3\beta}\cdot\frac{x+\sqrt{x^2-4y^2\cfrac{\cosh^3\beta}{(1+\cosh\beta)^2}}}{2y}\\
&<&
%\frac{\sinh\beta(1+\cosh\beta)}{\cosh^3\beta}\cdot\frac{x+x}{2y}\\
%&=&
\frac{\sinh\beta(1+\cosh\beta)}{\cosh^3\beta}\cdot\frac{x}{y},
\end{eqnarray*}
which is the desired inequality.
\end{proof}
Now, we are going to study an asymptotical behavior of the
trajectory of the dynamical system  \eqref{maindynsystem}.

\begin{theorem}\label{trajec1}
Let  $f:\Delta\to\Delta$ be the dynamical system given by
\eqref{maindynsystem}. Then the following assertions hold true:
\begin{enumerate}
\item[(i)] if $y^{(0)}>0$ then the trajectory
$\{(x^{(n)},y^{(n)})\}_{n=1}^{\infty}$ of $f$ starting from the
point $(x^{(0)},y^{(0)})$ is finite.

\item[(ii)] if $y^{(0)}=0$ then the trajectory
$\{(x^{(n)},y^{(n)})\}_{n=1}^{\infty}$  starting from the point
$(x^{(0)},y^{(0)})$ has the following form
\begin{equation*}
\left\{
\begin{array}{l}
x^{(n)} = \cfrac{\sqrt[2^n]{x^{(0)}\cosh^4\beta}}{\cosh^4\beta}\\
y^{(n)} = 0.
\end{array}
\right.
\end{equation*}
\end{enumerate}
\end{theorem}

\begin{proof} (i) Let $y^{(0)}>0$ and  suppose that the trajectory
$\{(x^{(n)},y^{(n)})\}_{n=1}^{\infty}$ of the dynamical system
starting from the point $(x^{(0)},y^{(0)})$ is infinite. This means
that the points $(x^{(n)},y^{(n)})$ are well defined for all
$n\in\bn.$ Then according to Remark \ref{confordomain} and Lemma
\ref{conditionforimage} we have
\begin{eqnarray}\label{conditionforxnandyn}
\frac{x^{(n)}}{y^{(n)}}<
\left(\frac{\sinh\beta(1+\cosh\beta)}{\cosh^3\beta}\right)^n\cdot\frac{x^{(0)}}{y^{(0)}}
\end{eqnarray} for all $n\in\bn.$

On the other hand, according to Remark \ref{confordomain}, $x^{(n)}$
and $y^{(n)}$ should satisfy the following inequality
\begin{eqnarray}\label{conforxnyn}
\frac{x^{(n)}}{y^{(n)}}\ge
2\sqrt{\frac{\cosh^3\beta}{(1+\cosh\beta)^2}},
\end{eqnarray}
for all $n\in\bn.$ Due to Lemma \ref{inequality} we find
$$\left(\frac{\sinh\beta(1+\cosh\beta)}{\cosh^3\beta}\right)^n\to
0 \ \  \textrm{as }  \ \  n\to\infty,$$  which with
\eqref{conditionforxnandyn} implies that the inequality
\eqref{conforxnyn} is not satisfied starting from some number
$N_0\in\bn.$ This contradiction shows that the trajectory
$\{(x^{(n)},y^{(n)})\}_{n=1}^{\infty}$ must be finite.\\

(ii) Now let $y^{(0)}=0$, then \eqref{maindynsystem} implies
$y^{(n)}=0$ for all $n\in\bn.$ Hence, from \eqref{maindynsystem} one
finds
$$x^{(n)}=\sqrt{\frac{x^{(n-1)}}{\cosh^4\beta}}.$$
So iterating last equality we obtain
$$x^{(n)}\cosh^4\beta=\sqrt[2^n]{x^{(0)}\cosh^4\beta},$$
which yields the desired equality.
\end{proof}

From the last Theorem \ref{trajec1}, we infer that the equation
\eqref{state} has a lot of parametrical solutions
$(w_0(\a),\{h_x(\a)\})$ given by
\begin{equation}\label{solutionofmainstate}
w_0(\alpha)=\left(
              \begin{array}{cc}
                \dfrac{1}{\alpha} & 0 \\
                0 & \dfrac{1}{\alpha} \\
              \end{array}
            \right),\quad
 h^{(n)}_x(\alpha)=\left(
                    \begin{array}{cc}
                      \dfrac{\sqrt[2^{n}]{\alpha\cosh^4\beta}}{\cosh^4\beta} & 0 \\
                      0 & \dfrac{\sqrt[2^{n}]{\alpha\cosh^4\beta}}{\cosh^4\beta} \\
                    \end{array}
                 \right),
                 \end{equation}
for every $x\in V.$, here $\alpha$ is any positive real number.

The boundary conditions corresponding to the fixed point of
\eqref{dynsystem} are the following ones:
\begin{equation}\label{solutionofmainstatewhenalphafixed}
w_0=\left(
               \begin{array}{cc}
                      {\cosh^4\beta} & 0 \\
                      0 & {\cosh^4\beta} \\
                    \end{array}
                  \right), \quad
                  h^{(n)}_x=\left(
                   \begin{array}{cc}
                \dfrac{1}{\cosh^4\beta} & 0 \\
                0 & \dfrac{1}{\cosh^4\beta} \\
              \end{array}
            \right), \ \ \forall x\in V,
            \end{equation}
which correspond to the value of $\alpha_0=\cfrac{1}{\cosh^4\beta}$
in \eqref{solutionofmainstate}. Therefore, further, we denote such
operators by $w_0\left(\alpha_0\right)$ and
$h_x^{(n)}\left(\alpha_0\right)$.

Let us consider the states $\ffi^{(n,f)}_{w_0(\a),\bh(\alpha)}$
corresponding to the solutions
$(w_0(\alpha),\{h_x^{(n)}(\alpha)\})$. By definition we have
\begin{eqnarray}\label{uniq-1}
\ffi^{(n,f)}_{w_0(\a),\bh(\alpha)}(x) &=&
\tr\left(w^{1/2}_{0}(\alpha)\prod_{i=0}^{n-1}K_{[i,i+1]}\prod_{x\in
\overrightarrow{W}_n}h^{(n)}_x(\alpha)
\prod_{i=1}^{n}K_{[n-i,n+1-i]}w^{1/2}_{0}(\alpha)x\right)\nonumber\\
&=&\frac{\left(\sqrt[2^{n+1}]{\alpha\cosh^4\beta}\right)^{2^{n+1}}}
{{\alpha}(\cosh^4\beta)^{2^{n+1}}}
\tr\left(\prod_{i=0}^{n-1}K_{[i,i+1]}\prod_{i=1}^{n}K_{[n-i,n+1-i]}x\right)\nonumber\\
&=&\frac{\alpha_0^{2^{n+1}}}{\alpha_0}\tr\left(\prod_{i=0}^{n-1}K_{[i,i+1]}\prod_{i=1}^{n}K_{[n-i,n+1-i]}x\right)\nonumber\\
&=&
\tr\left((w^{1/2}_{0}(\alpha_0)\prod_{i=0}^{n-1}K_{[i,i+1]}\prod_{x\in
\overrightarrow{W}_n}h^{(n)}_x(\alpha_0)
\prod_{i=1}^{n}K_{[n-i,n+1-i]}w^{1/2}_{0}(\alpha_0)x\right)\nonumber\\
&=&\ffi^{(n,f)}_{w_0(\a_0),\bh(\alpha_0)}(x),
\end{eqnarray}
for any $\alpha$. Hence, from the definition of forward QMC it
follows that
$\ffi^{(f)}_{w_0(\a),\bh(\alpha)}=\ffi^{(f)}_{w_0(\a_0),\bh(\alpha_0)}$,
which yields that the uniqueness of the forward QMC associated with
the model \eqref{1Kxy1}.  Therefore we have the following

\begin{theorem}\label{uni-3}
There is a unique forward QMC for the model \eqref{1Kxy1}.
\end{theorem}

Note that in \cite{arm} it was proved uniqueness of the ground state for the
one-dimensional quantum $XY$-model
\begin{equation}\label{xy-1d}
H=\b\sum_{n\in\bz}\{\s_{x}^{(n)}\s_{x}^{(n+1)}+\s_{y}^{(n)}\s_{y}^{(n+1)}\}.
\end{equation}
The proved Theorem \ref{uni-3} suggests that similar result can be
obtained for the Hamiltonian \eqref{xy-1d} in a Cayley tree of order
two.

{\bf Observation.} Let us denote
\begin{equation}\label{freeE1}
\tilde K_n(\a)=w^{1/2}_{0}(\alpha)\prod_{\{x,y\}\in E_1}K_{<x,y>}
\prod_{\{x,y\}\in E_2\setminus E_1}K_{<x,y>}\cdots \prod_{\{x,y\}\in
E_{n+1}\setminus E_{n}}K_{<x,y>}.
\end{equation}
Define a function $F(\b)$ by the following formula
\begin{eqnarray}\label{freeE1}
\b F(\b)=\lim_{n\to\infty}\frac{1}{|V_n|}\log \tr \left(\tilde
K_n^{*}(\a)\tilde K_n({\alpha})\right).
\end{eqnarray}
Using the same argument as in \eqref{uniq-1} one gets
\begin{eqnarray}\label{freeE2}
\frac{1}{|V_n|}\log \tr \left(\tilde K_n^{*}(\a)\tilde
K_n({\alpha})\right)&=&\frac{1}{|V_n|}\left(\log\bigg(
\frac{\sqrt[2^{n+1}]{\alpha\cosh^4\beta}}{\cosh^4\b}\bigg)^{-|W_{n+1}|}
+\log\big(\ffi^{(b)}_{w_0,\bh}(\id)\big)\right)\nonumber\\
&=&-\frac{|W_{n+1}|}{2^{n+1}|V_n|}\log\big(\alpha\cosh^4\beta\big)+
\frac{|W_{n+1}|}{|V_n|}\log\big(\cosh^4\beta\big).
\end{eqnarray}
So, taking into account
$\lim\limits_{n\to\infty}\frac{|W_{n+1}|}{|V_n|}=1$ from
\eqref{freeE2},\eqref{freeE1} we find
$$
F(\b)=\frac{4}{\b}\log\cosh\b.
$$
One can see that $F(\b)$ is an analytical function when $\b>0$. This
corresponds to the fact that the free energy of a system is an
analytical function. Of course, here the defined function $F$ is not
a free energy of the $XY$-model. On the other hand, for the same
model in a Cayley tree of order three we shall show the existence of
the quantum phase transition \cite{AMuS}. Moreover, it will be
established that derivative of certain thermodynamic function will
have discontinuity at critical values of $\b$.

\section{Conclusions}

Let us note that a first attempt of consideration of quantum Markov
fields began in \cite{[AcFi01a], [AcFi01b]} for the regular lattices
(namely for $\mathbb{Z}^d)$. But there, concrete examples of such
fields were not given. In the present paper, we have extended a
notion of QMC to fields, i.e. to Cayley tree. Note that such a tree
is the simplest hierarchical lattice with non-amenable graph
structure. This means that the ratio of the number of boundary sites
$W_n$ to the number of interior sites $V_n$ (see Sec. 2, for the
definitions of $W_n$ and $V_n$) of the tree tends to a nonzero
constant in the thermodynamic limit of a large system. Here QMCs
have been considered on discrete infinite tensor products of
$C^*$--algebras over such a tree. A tree structure of graphs allowed
us to give constructions of QMC, which generalizes the construction
of \cite{[AcFi03]} to trees. Namely, we have provided a construction
of a forward QMC defined on Cayley tree. By means of such a
construction we proved uniqueness of forward QMC associated with
$XY$-model on the second order Cayley tree. We have to stress here
that the constructed QMC associated with $XY$-model, is different
from thermal states of that model, since such states corresponds to
the $\exp(-\b \sum_{<x,y>}H_{<x,y>})$, which is different from a
product of $\exp(-\b H_{<x,y>})$. Roughly speaking, if we consider
the usual Hamiltonian system $H(\s)=-\b\sum_{<x,y>}h_{<x,y>}(\s)$,
then its Gibbs measure is defined by the fraction
\begin{equation}\label{mu-G1}
\mu(\s)=\frac{e^{-H(\s)}}{\sum_{\s}e^{-H(\s)}}.
\end{equation}
The such a measure can be viewed by another way as well. Namely,
\begin{equation}\label{mu-G2}
\mu(\s)=\frac{\prod_{<x,y>}e^{\b h_{<x,y>}(\s)}}{\sum\limits_{\s}\prod_{<x,y>}e^{\b h_{<x,y>}(\s)}}.
\end{equation}
A usual quantum mechanical definition of the quantum Gibbs states based on equation \eqref{mu-G1}. But our approach based on an alternative way (see \eqref{mu-G2}) of the definition of the quantum Gibbs states.
Note that whether or not the resulting states have a
physical interest is a question that cannot be solved on a
purely mathematical ground.

%\section*{  Bibliography }

\section*{Acknowledgement} The present study have been done within
the grant FRGS0308-91 of Malaysian Ministry of Higher Education. A finial part of
this work was done at the Abdus Salam International
Center for Theoretical Physics (ICTP), Trieste, Italy. F.M. thanks
the ICTP for providing financial support of his visit (within the scheme of
Junior Associate) to ICTP. The authors (F.M., M.S.) also acknowledge the Malaysian
Ministry of Scinece, technology and Innovation Grant 01-01-08-SF0079.

\appendix

\section{A proof of Lemma \ref{inequality}}\label{appx}

It is clear that, if $\beta >0,$ then
$$\sinh\beta\cosh\beta(1+\cosh\beta)>0.$$
Now we are going to show that
\begin{eqnarray}\label{rightineq}
\sinh\beta(1+\cosh\beta)<\cosh^3\beta.
\end{eqnarray}
Noting
$$\sinh\beta=\frac{e^{\beta}-e^{-\beta}}{2},\ \ \ \cosh\beta=\frac{e^{\beta}+e^{-\beta}}{2}.$$
and letting $t=e^{\beta},$  we reduce the last inequality
\eqref{rightineq} to
\begin{eqnarray}\label{powersix}
t^6-2t^5-t^4+7t^2+2t+1 &>& 0
\end{eqnarray}
Since $\beta >0,$ then $t>1$. Therefore, we shall show that
\eqref{powersix} is satisfied whenever $t>1$. Now consider several
cases with respect to $t$.

{\sc Case I.} Let $t\ge1+\sqrt{2}.$ Then we have
\begin{eqnarray*}
t^6-2t^5-t^4+7t^2+2t+1 =
t^4\big(t-(1+\sqrt{2})\big)\big(t-(1-\sqrt{2})\big)+7t^2+2t+1> 0
\end{eqnarray*}

{\sc Case II.} Let $2\le t\le 1+\sqrt{2}.$ Then it is clear that
$t<\sqrt{7}.$ Therefore,
\begin{eqnarray*}
t^6-2t^5-t^4+7t^2+2t+1 = t^5(t-2)+t^2(7-t^2)+2t+1> 0
\end{eqnarray*}

{\sc Case III.} Let $\sqrt{\frac72}\le t\le 2.$ Then one gets
\begin{eqnarray*}
2(t^6-2t^5-t^4+7t^2+2t+1)& = &2t^4\bigg(t^2-\frac72\bigg)+\frac52 t^4(2-t)\\
&&+\frac32 t^2(8-t^3)+2t^2+4t+2> 0
\end{eqnarray*}

{\sc Case IV.} Let $1<t\le\sqrt{\frac72}.$ Then we have
\begin{eqnarray*}
t^6-2t^5-t^4+7t^2+2t+1 = t^4(t-1)^2+t^2(7-2t^2)+2t+1> 0
\end{eqnarray*}

Hence, the inequality \eqref{rightineq} is satisfied for all
$\beta>0.$

%\end{proof}

\end{document}